\documentclass[11pt]{article}
\usepackage{enumerate}
\usepackage[top=1in,right=1in,left=1in,bottom=1in]{geometry}
\usepackage{latexsym, amssymb, amsmath, graphicx, amsthm}
\usepackage{cite}

\newtheorem{simplealgorithm}{{\bf Algorithm}}
\newtheorem{theorem}{Theorem}
\newtheorem{lemma}[theorem]{Lemma}
\newtheorem{proposition}[theorem]{Proposition}

\def\ind#1{\mathbf{1}_{#1}}
\def\esp#1{\mathbb{E}\left[#1\right]}

\def\CR{C\!R}
\def\CD{C\!D}
\def\sgn{\text{sgn}}

\title{A Tale of Two Metrics:\\ Simultaneous Bounds on Competitiveness and Regret}

\author{
Lachlan L. H. Andrew\thanks{Swinburne University of Technology} \\ \texttt{landrew@swin.edu.au}
\and
Siddharth Barman\thanks{California Institute of Technology} \\ \texttt{barman@caltech.edu}
\and
Katrina Ligett\footnotemark[2] \\ \texttt{katrina@caltech.edu}
\and
Minghong Lin\footnotemark[2] \\ \texttt{mhlin@caltech.edu}
\and
Adam Meyerson\thanks{Google, Inc.} \\ \texttt{awmeyerson@gmail.com}
\and
Alan Roytman\thanks{Tel-Aviv University} \\ \texttt{alan.roytman@cs.tau.ac.il}
\and
Adam Wierman\footnotemark[2] \\ \texttt{adamw@caltech.edu}
}
\date{}

\begin{document}

\maketitle

\begin{abstract}
We consider algorithms for ``smoothed online convex optimization''
problems, a variant of the class of online convex optimization
problems that is strongly related to metrical task
systems.  Prior literature on these problems has focused on two
performance metrics: regret and the competitive ratio.  There exist known
algorithms with sublinear regret and known algorithms with constant competitive
ratios; however, no known algorithm achieves both simultaneously.  We show that this
is due to a fundamental incompatibility between these two metrics -- no algorithm
(deterministic or randomized) can achieve sublinear regret and a constant competitive
ratio, even in the case when the objective functions are linear.  However, we also
exhibit an algorithm that, for the important special case of one-dimensional decision
spaces, provides sublinear regret while maintaining a competitive ratio that grows
arbitrarily slowly.
\end{abstract}

\section{Introduction}
In an \emph{online convex optimization} (OCO) problem, a learner
interacts with an environment in a sequence of rounds.  During each
round $t$: (i) the learner must choose an action $x^t$ from a convex
decision space $F$; (ii) the environment then reveals a convex cost
function $c^t$, and (iii) the learner experiences cost $c^t(x^t)$.
The goal is to design learning algorithms that minimize the cost
experienced over a (long) horizon of $T$ rounds.

In this paper, we study a generalization of online convex optimization
that we term \emph{smoothed online convex optimization} (SOCO).  The
only change in SOCO compared to OCO is that the cost experienced by
the learner each round is now $c^t(x^t) + \|x^t-x^{t-1}\|$, where $\| \cdot \|$
is a seminorm.\footnote{Recall that a seminorm satisfies
the axioms of a norm except that $\|x\|=0$ does not imply
$x=0$.}
That is, the learner experiences a ``smoothing
cost'' or ``switching cost'' associated with changing the action, in addition
to the ``operating cost'' $c(\cdot)$.

Many applications typically modeled using online convex
optimization have, in reality, some cost associated with a change of
action.  For example, switching costs in the $k$-armed bandit setting
have received considerable attention \cite{Asawa96,Guha09}.  Additionally, a strong
motivation for studying SOCO comes from the recent developments in dynamic
capacity provisioning algorithms for data centers
\cite{kusic2007risk,LWAT11,Lin12,urgaonkar2011optimal,Minghua_CSR,zhang2012dynamic,yang2012dynamic},
where the goal is to dynamically control the number and placement of
active servers ($x^t$) in order to minimize a combination of the delay
and energy costs (captured by $c^t$) and the switching costs involved
in cycling servers into power saving modes and migrating data
($\|x^t-x^{t-1}\|$). Further,
SOCO has applications even in contexts where there are no costs
associated with switching actions.  For example, if there is concept
drift in a penalized estimation problem, it is natural to make
use of a regularizer (switching cost) term in order to control the
speed of the drift of the estimator.

\paragraph{Two communities, two performance metrics.}
Though the precise formulation of SOCO does not appear to have been
studied previously, there are two large bodies of literature on
closely related problems: (i) the online convex optimization (OCO)
literature within the machine learning community, e.g.,
\cite{Zinkevich03,Hazan07}, and (ii) the metrical task system (MTS)
literature within the algorithms community, e.g., \cite{BLS92,MGS88}.
We discuss these literatures in detail in Section
\ref{sec:background}.  While there are several differences between the
formulations in the two communities, a notable difference is that they
focus on different performance metrics.

In the OCO literature, the goal is typically to
minimize the \textbf{{\em regret}}, which is the
difference between the cost of the algorithm and the cost of the
offline optimal static solution.  In this context, a number of
algorithms have been shown to provide sublinear regret (also called ``no
regret''). For example, online gradient descent can achieve
$O(\sqrt{T})$-regret \cite{Zinkevich03}.  Though such guarantees are
proven only in the absence of switching costs, we show in Section
\ref{sec:oco} that the same regret bound also holds for SOCO.

In the MTS literature, the goal is typically to
minimize the \textbf{{\em competitive ratio}}, which is the maximum
ratio between the cost of the algorithm and the cost of the offline
optimal (dynamic) solution.  In this setting, most results tend to be
``negative,'' e.g., when $c^t$ are arbitrary, for any metric space the competitive ratio of an
MTS algorithm with states chosen from that space grows without bound as the
number of states grows \cite{BLS92,Blum92}.  However,
{these results still yield competitive ratios that do not increase with
the number of tasks, i.e., with time.  Throughout, we neglect
dependence of the competitive ratio on the number of states, and describe the
competitive ratio as ``constant'' if it does not grow with time.
Note also that}
positive results have emerged when
the cost function and decision space are convex~\cite{LWAT11}.

Interestingly, the focus on different performance metrics in the OCO
and MTS communities has led the communities to develop quite different
styles of algorithms.  The differences between the algorithms is
highlighted by the fact that \emph{all algorithms developed in the OCO
community have poor competitive
ratio and all algorithms developed in the MTS community
have poor regret}.

However, it is natural to seek algorithms with both low regret and low competitive ratio.
In learning theory, doing well for both corresponds to being able to learn both static and dynamic concepts well.
In the design of a dynamic controller, low regret shows that the control is not more risky than
static control, whereas low competitive ratio shows that the control is nearly as good as the best dynamic controller.

The first to connect the two metrics were
\cite{Blum00}, who treated the special case
where the switching costs are a fixed constant, instead of a norm.  In this
context, they showed how to translate bounds on regret to
bounds on the competitive ratio, and vice versa.
A recent breakthrough was made by
\cite{BCNS12} who used a primal-dual approach to develop an
algorithm that performs well for the ``$\alpha$-unfair competitive
ratio,'' which is a hybrid of the competitive ratio and regret that
provides comparison to the dynamic optimal when $\alpha=1$ and to the
static optimal when $\alpha=\infty$ (see Section~\ref{sec:formulation}).
Their algorithm was not shown to perform well \emph{simultaneously} for regret and the
competitive ratio, but the result highlights the feasibility of
unified approaches for algorithm design across competitive ratio and
regret.\footnote{There is also work
on achieving simultaneous guarantees with respect to the static and
dynamic optima in other settings, e.g., decision making on
lists and trees \cite{blum2002static}, and there have been
applications of algorithmic approaches from machine learning to
MTS \cite{blum1999finely,abernethy2010regularization}.}

\paragraph{Summary of contributions.}
This paper explores the relationship between
minimizing regret and minimizing the
competitive ratio.  To this end, we seek to answer the following question: ``Can an algorithm simultaneously achieve both a constant competitive ratio and a sublinear regret?''

To answer this question,
we show that \emph{there is a fundamental incompatibility
between regret and competitive ratio} --- no algorithm can maintain
both sublinear regret and a constant competitive ratio (Theorems
\ref{thm:impossibility}, \ref{thm:impossibility_online}, and \ref{th:CR1R1}).  This ``incompatibility'' does not stem from
a pathological example: it holds even for the simple case when $c^t$
is linear and $x^t$ is scalar.  Further, it holds for both
deterministic and randomized algorithms and also when the $\alpha$-unfair
competitive ratio is considered.

Though providing both sublinear regret and a constant competitive ratio is
impossible, we show that it is possible to ``nearly'' achieve this
goal.  We present an algorithm,
\emph{``Randomly Biased Greedy'' (RBG), which achieves a competitive ratio
of $(1+\gamma)$ while maintaining $O(\max\{T/\gamma,\gamma\})$ regret} for
$\gamma\ge 1$ on one-dimensional action spaces.  If $\gamma$ can be
chosen as a function of $T$, then this algorithm can
balance between regret and the competitive ratio.  For example, it can
achieve sublinear regret while having an arbitrarily slowly growing
competitive ratio, or it can achieve $O(\sqrt{T})$ regret while
maintaining an $O(\sqrt{T})$ competitive ratio.
Note that, unlike the scheme of~\cite{BCNS12}, this algorithm has
a finite competitive ratio on continuous action spaces and provides
a \emph{simultaneous} guarantee on both regret and the competitive
ratio.

\section{Problem Formulation}\label{sec:formulation}
An instance of smoothed online convex optimization (SOCO) consists of a convex decision\slash action
space $F\subseteq (\mathbb{R}^+)^n$ and a sequence of cost functions $\{c^1,c^2,\dots\}$, where each $c^t:F\to\mathbb{R}^+$.
At each time $t$, a learner\slash algorithm chooses an action vector $x^t \in F$ and the environment chooses a cost function $c^t$.
Define the \emph{$\alpha$-penalized cost with lookahead $i$} for the sequence $\dots, x^t, c^t, x^{t+1}, c^{t+1}, \dots$ to be
$$
    C_i^\alpha(A,T)=\mathbb{E}\left[ \sum_{t=1}^T c^{t}(x^{t+i})+\alpha\|x^{t+i}-x^{t+i-1}\|\right],
$$
where $x^1,\ldots,x^T$ are the decisions of algorithm $A$, the initial action is
$x^{i}=0$ without loss of generality,
the expectation is over randomness in the algorithm, and
$\|\cdot\|$ is a seminorm on $\mathbb{R}^n$.  We usually suppress the parameter $T$.

In the OCO and MTS literatures the learners pay different special cases of this cost.
In OCO the algorithm ``plays first'' giving a 0-step
lookahead and switching costs are ignored, yielding $C_0^0$.  In MTS the environment plays first
giving the algorithm 1-step lookahead ($i = 1$), and uses $\alpha=1$, yielding $C_1^1$.
Note that we sometimes omit the superscript when $\alpha=1$, and the subscript when $i=0$.

One can relate the MTS and OCO costs by relating
$C^\alpha_i$ to $C^\alpha_{i-1}$,
as done by \cite{Blum00} and~\cite{BCNS12}.  The penalty due to not having lookahead is
\begin{align} \label{eq:cost_for_offset}
c^t(x^t)-c^t(x^{t+1})
 & \le \triangledown c^t(x^t)(x^{t}-x^{t+1})
 \le \|\triangledown c^t(x^t)\|_2 \cdot \|x^t-x^{t+1}\|_2,
\end{align}
where $\|\cdot\|_2$ is the Euclidean norm.
We adopt the assumption, common in the OCO literature, that $\|\triangledown c^t(\cdot)\|_2$ are bounded on a given instance; which thus bounds the difference
between the costs of MTS and OCO (with switching cost), $C_1$ and $C_0$.

\paragraph{Performance metrics.}
The performance of a SOCO algorithm is typically evaluated by comparing its cost to that of an offline ``optimal'' solution, but the communities
differ in their choice of benchmark, and whether to compare additively or multiplicatively.

The OCO literature typically compares against the optimal offline \emph{static} action, i.e.,
$$OPT_s = \min_{x\in F} \sum_{t=1}^T c^t(x),$$
and evaluates the \textbf{\emph{regret}}, defined as the
(additive) difference between the algorithm's cost and that of the optimal static
action vector. Specifically, the regret $R_i(A)$ of Algorithm $A$ with lookahead $i$ on instances $\mathfrak{C}$,
is less than $\rho(T)$ if for any sequence of cost functions $(c^1,\ldots,c^T)\in\mathfrak{C}^T$,
\begin{equation}\label{eq:regret}
    C^0_i(A)-OPT_s \leq \rho(T).
\end{equation}
Note that for any problem and any $i\ge 1$ there exists an algorithm $A$ for which $R_i(A)$ is non-positive; however,
an algorithm that is not designed specifically to minimize regret may have $R_i(A) > 0$.

This traditional definition of regret omits switching costs and
lookahead (i.e., $R_0(A)$).  One can generalize regret to define $R'_i(A)$, by
replacing $C^0_i(A)$ with $C^1_i(A)$ in Equation~\eqref{eq:regret}.
Specifically, $R'_i(A)$
is less than $\rho(T)$ if for any sequence of cost functions $(c^1,\ldots,c^T)\in\mathfrak{C}^T$,
\begin{equation*}
    C^1_i(A)-OPT_s \leq \rho(T).
\end{equation*}
Except where noted, we use the set $\mathfrak{C}^1$ of sequences of convex functions mapping
$(\mathbb{R}^+)^n$ to $\mathbb{R}^+$ with (sub)gradient uniformly bounded over the sequence.
Note that we do not require differentiability; throughout this paper, references to gradients can be read as references to subgradients.

The MTS literature typically compares against the optimal offline (dynamic) solution,
$$
  OPT_d = \min_{x\in F^T}\sum_{t=1}^T c^{t}(x^{t})+  \|x^{t}-x^{t-1}\|,
$$
and evaluates the \textbf{\emph{competitive ratio}}.   The cost most commonly considered is $C_1$. More
generally, we say the competitive ratio with lookahead $i$, denoted by $CR_i(A)$, is $\rho(T)$ if for any
sequence of cost functions $(c^1,\ldots,c^T)\in\mathfrak{C}^T$,
\begin{equation}\label{eq:CR}
    C_i(A) \leq \rho(T) OPT_d + O(1).
\end{equation}

\paragraph{Bridging competitiveness and regret.}
Many hybrid benchmarks have been proposed to bridge static and dynamic comparisons.
For example, Adaptive-Regret
\cite{hazan2009efficient} is the maximum regret over any interval, where the ``static'' optimum can differ for different intervals, and
internal regret~\cite{blum2005internalregret} compares the online policy against a simple perturbation of that policy. We adopt the static-dynamic hybrid proposed in the most closely related literature~\cite{Blum92,Blum00,BCNS12},
the \textbf{\emph{$\alpha$-unfair competitive ratio}},  which we denote by $CR_i^\alpha(A)$ for lookahead $i$.  For
$\alpha\geq 1$,
$CR_i^\alpha(A)$ is $\rho(T)$ if Equation~\eqref{eq:CR} holds with $OPT_d$ replaced by
$$
    OPT_d^\alpha = \min_{x\in F^T}\sum_{t=1}^T c^{t}(x^{t})+ \alpha\|x^{t}-x^{t-1}\|.
$$
Specifically, the $\alpha$-unfair competitive ratio with lookahead $i$, $CR_i^\alpha(A)$, is $\rho(T)$ if for any
sequence of cost functions $(c^1,\ldots,c^T)\in\mathfrak{C}^T$,
\begin{equation*}
    C_i(A) \leq \rho(T) OPT_d^\alpha + O(1).
\end{equation*}
For $\alpha=1$, $OPT_d^\alpha$ is the dynamic optimum; as $\alpha \rightarrow\infty$, $OPT_d^\alpha$ approaches the static optimum.

To bridge the additive versus multiplicative comparisons used in the two literatures, we define the \textbf{\emph{competitive difference}}.  The $\alpha$-unfair competitive difference with lookahead $i$ on
instances $\mathfrak{C}$, $CD_i^\alpha(A)$, is $\rho(T)$ if for any sequence of cost functions $(c^1,\ldots,c^T)\in\mathfrak{C}^T$,
\begin{equation*}
    C_i(A) - OPT^\alpha_{d} \leq \rho(T).
\end{equation*}
This measure allows for a smooth transition between regret (when $\alpha$ is large enough) and an additive version of the competitive ratio when $\alpha=1$.

\section{Background} \label{sec:background}
In the following, we briefly discuss related work
on both online convex optimization and metrical task systems,
to provide context for the results in the current paper.

\subsection{Online Convex Optimization}\label{sec:oco}
The OCO problem has a rich history and a wide range of important applications.  In computer science,
OCO is perhaps most associated with
the $k$-experts problem \cite{Herbster98,Littlestone94}, a
discrete-action version of online optimization wherein at each round
$t$ the learning algorithm must choose a number between 1 and $k$, which can be viewed as following the advice of one of $k$ ``experts.''
However, OCO also has a long history in other areas,
such as portfolio management \cite{Cover91,calafiore2008multi}.

The identifying features of the OCO formulation are that (i)~the typical performance metric
considered is regret, (ii)~switching costs are not considered, and (iii)~the
learner must act before the environment reveals the cost function.  In our notation, the
cost function in the OCO literature is $C^0(A)$ and the performance
metric is $R_0(A)$.  Following \cite{Zinkevich03} and~\cite{Hazan07},
the typical assumptions are that the decision space $F$ is non-empty, bounded and closed, and that
the Euclidean norms of gradients $\|\triangledown c^t(\cdot)\|_2$ are also bounded.

In this setting, a number of algorithms have been shown to achieve ``no regret,''
i.e., sublinear regret, $R_0(A)=o(T)$. An important example
is {\em online gradient descent} (OGD), which is parameterized by learning rates
$\eta_t$. OGD works as follows.
\begin{simplealgorithm}[Online Gradient Descent, OGD]
Select an arbitrary $x^1 \in F$. At time step $t \geq 1$, select $x^{t+1}=P(x^t-\eta_t \triangledown c^t(x^t))$, where $P(y)=\arg\min_{x\in F} \|x-y\|_2$
is the projection under the Euclidean norm.
\end{simplealgorithm}
With appropriate learning rates $\eta_t$, OGD achieves sublinear regret for OCO. In particular, the variant
of \cite{Zinkevich03} uses $\eta_t=\Theta(1/\sqrt{t})$ and obtains $O(\sqrt{T})$-regret.  Tighter bounds
are possible in restricted settings. The work of~\cite{Hazan07} achieved $O(\log T)$
regret by choosing $\eta_t=1/(\gamma t)$ for settings when the cost function additionally
is twice differentiable and has minimal curvature, i.e., $\triangledown^2 c^t(x) - \gamma I_n$ is
positive semidefinite for all $x$ and $t$, where $I_n$ is the identity matrix of size $n$.
In addition to algorithms based on gradient descent, more recent algorithms such as Online Newton
Step and Follow the Approximate Leader \cite{Hazan07} also attain $O(\log T)$-regret bounds for a
class of cost functions.

None of the work discussed above considers switching costs.
To extend the literature discussed above from OCO to SOCO, we need to track the switching costs incurred
by the algorithms.
This leads to the following straightforward result, proven in Appendix~\ref{sec:proof.ocotosoco}.

\begin{proposition} \label{P.OCOTOSOCO}
Consider an online gradient descent algorithm $A$ on a finite dimensional space with learning rates such
that $\sum_{t=1}^T \eta_t=O(\rho_1(T))$.  If $R_0(A)=O(\rho_2(T))$, then we have $R'_0(A)=O(\rho_1(T)+\rho_2(T))$.
\end{proposition}
Interestingly, the choices of $\eta_t$ used by the algorithms designed for OCO also turn out to be good choices to control the switching costs of the algorithms.
The algorithms of \cite{Zinkevich03} and~\cite{Hazan07}, which use $\eta_t=1/\sqrt{t}$ and $\eta_t=1/(\gamma t)$, are still $O(\sqrt{T})$-regret
and $O(\log T)$-regret respectively when switching costs are considered, since in these cases $\rho_1(T) = O(\rho_2(T))$.  Note that  a similar
result can be obtained for Online Newton Step \cite{Hazan07}.

Importantly, though the regret of OGD algorithms is sublinear, it can easily be
shown that the competitive ratio of these algorithms is unbounded.

\subsection{Metrical Task Systems}
Like OCO, MTS also has a rich history and a wide range of important
applications.  Historically, MTS is perhaps most
associated with the $k$-server problem \cite{CMP08}.  In this problem, there are $k$ servers, each in
some state, and a sequence of requests is incrementally revealed. To serve a request, the system
must move one of the servers to the state necessary to serve the request, which incurs a cost that depends on the source and destination states.

The formulation of SOCO in Section \ref{sec:formulation} is actually, in many ways, a special case of the most general MTS formulation.  In general, the MTS formulation differs in that (i) the cost functions $c^t$ are not assumed to be convex, (ii) the decision space is typically assumed to be discrete and is not necessarily embedded in a vector space,
and (iii) the switching cost is an arbitrary metric $d(x^t, x^{t-1})$ rather than a seminorm $\|x^t - x^{t-1}\|$.  In this context, the cost function studied by MTS is typically $C_1$ and the performance metric of interest is the competitive ratio, specifically $CR_1(A)$, although the $\alpha$-unfair competitive ratio $CR_1^\alpha$ also receives attention.

The weakening of the assumptions on the cost functions, and the fact that the competitive ratio uses the dynamic optimum as the benchmark, means that
most of the results in the MTS setting are ``negative'' when compared with those for OCO.  In particular, it has been proven that, given an arbitrary
metric decision space of size $n$, any deterministic algorithm must be $\Omega(n)$-competitive \cite{BLS92}.  Further, any randomized algorithm must
be $\Omega(\sqrt{\log n/ \log\log n})$-competitive \cite{Blum92}.

These results motivate imposing additional structure on the cost functions to attain positive results.  For example, it is commonly assumed that the metric
is the uniform metric, in which $d(x,y)$ is equal for all $x\neq y$; this assumption was made by~\cite{Blum00} in a study of the tradeoff between competitive
ratio and regret.  For comparison with OCO, an alternative natural
restriction is to impose convexity assumptions on the cost function and the decision space, as done in this paper.

Upon restricting $c^t$ to be convex, $F$ to be convex, and $\|\cdot\|$ to be a semi-norm, the MTS formulation becomes quite similar to the SOCO formulation.
This restricted class has been the focus of a number of recent papers, and some positive results have emerged.  For example,
\cite{LWAT11} showed that when $F$ is a one-dimensional normed
space,\footnote{We need only consider the absolute value norm,
since for every seminorm $\|\cdot\|$ on $\mathbb{R}$, $\|x\|=\|1\| |x|$.} a deterministic online algorithm called Lazy Capacity Provisioning (LCP) is $3$-competitive.

Importantly, though the algorithms described above provide constant competitive ratios, in all cases it is easy to see that the regret of these algorithms is linear.

\section{The Incompatibility of Regret and the Competitive Ratio}\label{sec:incompatibility}
As noted in the introduction, there is considerable motivation to perform well for regret and competitive ratio simultaneously, see also
\cite{Blum92,Blum00,BCNS12,hazan2009efficient,blum2005internalregret}.
None of the algorithms discussed so far achieves this goal.  For example, Online Gradient Descent has sublinear regret but its competitive ratio is infinite.
Similarly, Lazy Capacity Provisioning is 3-competitive but has linear regret.

This is no accident. We show below that the two goals are fundamentally incompatible: any algorithm that has
sublinear regret for OCO necessarily has an infinite competitive ratio for MTS; and any algorithm that has a constant competitive ratio for MTS necessarily has at least
linear regret for OCO.  Further, our results give lower bounds on the simultaneous guarantees that are possible.

In discussing this ``incompatibility,'' there are a number of subtleties as a result of the differences in formulation between the OCO literature, where regret is the
focus, and the MTS literature, where competitive ratio is the focus.  In particular, there are four key differences which are important to highlight: (i) OCO uses
lookahead $i=0$ while MTS uses $i=1$; (ii) OCO does not consider switching costs ($\alpha=0$) while MTS does ($\alpha=1$); (iii) regret uses an additive comparison while
the competitive ratio uses a multiplicative comparison; and (iv) regret compares to the static optimal while competitive ratio compares to
the dynamic optimal.  Note that the first two are intrinsic to the costs, while the latter are intrinsic to the performance metric.
The following teases apart which of these differences create incompatibility and which do not.  In particular, we prove that (i) and (iv) each create incompatibilities.

Our first result in this section states that there is an incompatibility between regret in the OCO setting and the competitive ratio in the MTS setting, i.e., between
the two most commonly studied measures $R_0(A)$ and $CR_1(A)$.  Naturally, the incompatibility remains if switching costs are added to regret, i.e.,
$R'_0(A)$ is considered.  Further, the incompatibility remains when the competitive difference is considered, and so both the comparison with
the static optimal and the dynamic optimal are additive.  In fact, the incompatibility remains even when the $\alpha$-unfair competitive
ratio/difference is considered.  Perhaps most surprisingly, the incompatibility remains when there is lookahead, i.e., when $C_i$ and $C_{i+1}$ are considered.

\begin{theorem}\label{thm:impossibility}
Consider an arbitrary seminorm $\|\cdot\|$ on $\mathbb{R}^n$, constants $\gamma>0$, $\alpha\geq 1$,
and $i\in \mathbb{N}$. There is a $\mathfrak{C}$ containing a single sequence of cost functions such that, for all deterministic and randomized algorithms $A$,
either $R_i(A)=\Omega(T)$ or, for large enough $T$, both $\CR^\alpha_{i+1}(A)\ge \gamma$ and $\CD^\alpha_{i+1}(A)\ge \gamma T$.
\end{theorem}
The incompatibility arises even in ``simple'' instances; the proof of Theorem \ref{thm:impossibility} uses linear cost functions and a one-dimensional decision space,
and the construction of the cost functions does not depend on $T$ or $A$.

The cost functions used by regret and the competitive ratio in Theorem \ref{thm:impossibility} are ``off by one,'' motivated by the different settings in OCO and MTS.  However, the following shows that parallel results also hold when the cost functions are not ``off by one,'' i.e., for $R_0(A)$ versus $CR^\alpha_0(A)$ and $R'_1(A)$ versus $CR^\alpha_1(A)$.

\begin{theorem}\label{thm:impossibility_online}
Consider an arbitrary seminorm $\|\cdot\|$ on $\mathbb{R}^n$,
constants $\gamma>0$ and $\alpha\geq 1$,
and a deterministic or randomized online algorithm $A$.
There is a $\mathfrak{C}$ containing two cost functions such that
either $R_0(A)=\Omega(T)$ or, for large enough $T$, both $\CR^\alpha_{0}(A)\ge \gamma$ and $\CD^\alpha_{0}(A)\ge \gamma T$.
\end{theorem}

\begin{theorem}\label{th:CR1R1}
Consider an arbitrary norm $\|\cdot\|$ on $\mathbb R^n$.
There is a $\mathfrak{C}$ containing two cost functions such that,
for any constants $\gamma>0$
and $\alpha\geq 1$ and any deterministic or randomized online algorithm $A$,
either $R'_1(A)=\Omega(T)$ or, for large enough $T$, $\CR_1^\alpha(A) \ge \gamma$.
\end{theorem}
The impact of these results can be stark.  It is impossible for an algorithm to learn static concepts with sublinear regret in the OCO setting, while having a constant
competitive ratio for learning dynamic concepts in the MTS setting.  More strikingly, in control theory, any dynamic controller that has a constant
competitive ratio must have at least linear regret, and so there are cases where it does much worse than the best static
controller.  Thus, one cannot simultaneously guarantee the dynamic policy is always as good as the best static policy and is nearly as good as the optimal dynamic policy.

Theorem \ref{th:CR1R1} is perhaps the most interesting of these results.  Theorem \ref{thm:impossibility} is due to seeking to minimize different cost functions
($c^t$ and $c^{t+1}$), while Theorem \ref{thm:impossibility_online} is due to the hardness of attaining a
small $CR_0^\alpha$, i.e., of mimicking the dynamic optimum without 1-step lookahead.  In contrast, for Theorem~\ref{th:CR1R1}, algorithms exist with strong performance
guarantees for each measure individually, and the measures are aligned in time.
However, Theorem \ref{th:CR1R1} must consider the (nonstandard) notion of regret that includes switching costs ($R'$), since otherwise the problem is trivial.

\subsection{Proofs}
We now prove the results above.  We use one-dimensional examples; however, these examples can easily be embedded into higher dimensions if desired.
We show proofs only for competitive ratio; the proofs for competitive difference are similar.

Let $\bar{\alpha} = \max(1,\|\alpha\|)$.  Given $a>0$ and $b\ge0$, define two possible cost functions on
$F=[0,1/\bar{\alpha}]$: $f_1^\alpha(x) = b + ax\bar{\alpha}$ and
$f_2^\alpha(x) = b + a(1-x\bar{\alpha})$.
These functions are similar to those used by~\cite{Gur2012} to study online gradient descent to learn a concept of bounded total variation.
To simplify notation, let $D(t)=1/2 - \esp{x^t}\bar\alpha$, and note that $D(t)\in [-1/2,1/2]$.

\subsubsection{Proof of Theorem~\ref{thm:impossibility}}
To prove Theorem~\ref{thm:impossibility}, we prove the stronger claim that
$CR^\alpha_{i+1}(A) + {R_i(A)/T} \ge \gamma$.

Consider a system with costs $c^t = f_1^\alpha$ if $t$ is odd and $f_2^\alpha$ if $t$ is even.  Then $C_i(A) \ge (a/2+b)T + a \sum_{t=1}^T (-1)^t D(t+i)$.
The static  optimum is not worse than the scheme that sets $x^t=1/(2\bar{\alpha})$ for all $t$, which has total cost no more than $(a/2+b)T + \|1/2\|$.
The $\alpha$-unfair dynamic optimum for $C_{i+1}$ is not worse than the scheme that sets
$x^t=0$ if $t$ is odd and $x^t=1/\bar{\alpha}$ if $t$ is even, which
has total $\alpha$-unfair cost at most $(b+1)T$.  Hence
\begin{align*}
R_i(A) \ge& a\sum_{t=1}^T (-1)^t D(t+i) - \|1/2\|, \\
\CR_{i+1}^\alpha(A) \ge& \frac{(a/2+b)T + a\sum_{t=1}^T (-1)^t D(t+i+1)}{(b+1)T}.
\end{align*}
Thus, since $D(t)\in [-1/2,1/2]$,
\begin{align*}
    (b+1)T(\CR_{i+1}^{\alpha}(A) + R_{i}(A)/T) &+ (b+1)\|1/2\| - (a/2+b)T\\
&\geq
    a\sum_{t=1}^T (-1)^t (D(t+i+1) + (b+1) D(t+i))\\
&=
    ab\sum_{t=1}^T (-1)^t D(t+i)
    -a \left(D(i+1) + (-1)^T D(T+i+1)\right) \\
&\geq
	- abT/2 - a.
\end{align*}
To establish the claim, it is then sufficient that $(a/2 + b)T - (b+1)\|1/2\| - abT/2 - a \ge \gamma T (b + 1)$.
For $b=1/2$ and $a=30\gamma + 2 + \|6\|$, this holds for $T \ge 5$.

\subsubsection{Proof of Theorem~\ref{thm:impossibility_online}}
To prove Theorem~\ref{thm:impossibility_online}, we again prove the stronger claim
$CR^\alpha_{0}(A) + {R_0(A)/T} \ge \gamma$.

Consider the cost function
sequence with $c^t(\cdot)=f_2^0$ for $\esp{x^t}\le 1/2$ and $c^t(\cdot)=f_1^0$
otherwise, on decision space $[0,1]$, where $x^t$ is the (random) choice of the algorithm at round $t$.
Here the expectation is taken over the marginal distribution of $x^t$ conditioned on
$c_1,\dots,c_{t-1}$, averaging out the dependence on the realizations of $x_1,\dots,x_{t-1}$.
Notice that this sequence can be constructed by an oblivious adversary before
the execution of the algorithm.

The following lemma is proven in Appendix~\ref{sec:proof.rgbound}.

\begin{lemma} \label{LM:RGBOUND}
Given any algorithm, the sequence of cost functions chosen by the above oblivious adversary gives the following:
\begin{align} \label{eq:rgbound}
R_0(A),  R'_0(A)&\ge a\sum_{t=1}^T |1/2-\esp{x^t}| - \|1/2\|, \\
\label{eq:crbound}
  \CR^\alpha_0(A) &\ge \frac{(a/2+b)T+a\sum_{t=1}^T
		       |1/2-\esp{x^{t}}|}{(b+\|\alpha\|)T}.
\end{align}
\end{lemma}
From Equation~\eqref{eq:rgbound} and Equation~\eqref{eq:crbound} in Lemma~\ref{LM:RGBOUND}, we have
$\CR^\alpha_0(A) + R_0(A)/T \ge \frac{(a/2+b)T}{(b+\|\alpha\|)T} - \frac{\|1/2\|}{T}$.
For $a > 2\gamma(b+\|\alpha\|)$, the right hand side is bigger than $\gamma$ for sufficiently large $T$,
which establishes the claim.

\subsubsection{Proof of Theorem~\protect\ref{th:CR1R1}}
Let $a=\|1\|/2$ and $b=0$.
Let $M = 4\alpha\gamma \|1\|/a= 8\alpha\gamma$.  For $T\gg M$,
divide $[1,T]$ into segments of length $3M$.  For the last $2M$ of each
segment, set $c^t=f_1^\alpha$.  This ensures that the static optimal solution is $x=0$.  Moreover, if $c^t$ is either $f_1^\alpha$ or $f_2^\alpha$
for all $t$ in the first $M$ time steps, then the optimal dynamic solution is also $x^t=0$ for the last $2M$ time steps.

Consider a solution for which each segment has non-negative regret.  Then to
obtain sublinear regret, for any positive threshold $\epsilon$, at least
$T/(3M)-o(T)$ of these segments must have regret below
$\epsilon\|1/\bar{\alpha}\|$.  We then show that these segments must have high competitive ratio.
To make this more formal, consider (without loss of generality) the single segment $[1,3M]$.

Let $\tilde{c}$ be such that $\tilde{c}^t = f_2^\alpha$ for all $t\in [1,M]$ and $\tilde{c}^t=f_1^\alpha$ for $t>M$.  Then
the optimal dynamic solution on $[1,3M]$ is $x_d^t =
\ind{t \le M}/\bar{\alpha}$, which has total cost $2\alpha\|1/\bar{\alpha}\|$
consisting entirely of switching costs.

The following lemma is proven in Appendix~\ref{sec:CR1R1lemma}.
\begin{lemma}\label{th:minRegret}
For any $\delta\in(0,1/\bar{\alpha})$ and integer $\tau>0$, there exists an
$\epsilon(\delta,\tau)>0$ such that, if
$c^t=f_2^\alpha$ for all $1\le t \le \tau$ and $x^t > \delta$ for any $1 \le t \le
\tau$,
then there exists an $m \le \tau$ such that
$C_1(x, m) -C_1(OPT_s,m)> \epsilon(\delta,\tau) \|1/\bar{\alpha}\|$.
\end{lemma}

Let $\delta = 1/[5 \bar\alpha] \in (0,1)$.  For any decisions such that
$x^t < \delta$ for all $t\in [1,M]$, the operating cost of $x$ under $\tilde{c}$ is at
least $3\alpha\gamma \|1/\bar{\alpha}\|$.
Let the adversary choose a $c$ on this segment such that $c^t=f_2^\alpha$ until (a)
the first time $t_0<M$ that the algorithm's solution $x$ satisfies $C_1(x, t_0)
-C_1(OPT_s,t_0) > \epsilon(\delta,M)\|1/\bar{\alpha}\|$, or (b) $t=M$.  After this, it chooses $c^t=f_1^\alpha$.

In case (a), $C_1(x,3M) - C_1(OPT_s,3M) > \epsilon(\delta,M)\|1/\bar{\alpha}\|$
by Lemma~\ref{th:minRegret}, since $OPT_s$ incurs no cost after $t_0$.  Moreover $C_1(x,3M) \ge C_1(OPT_d,3M)$.

In case (b), $C_1(x,3M) / C_1(OPT_d,3M) \ge 3\alpha\gamma\|1/\bar{\alpha}\| /
(2\alpha \|1/\bar{\alpha}\|) =3\gamma/2$.

To complete the argument, consider all segments.
Let $g(T)$ be the number of segments for which case (a) occurs.
The regret then satisfies
$$
    R'_1(A) \ge \epsilon(\delta,M)\|1/\bar{\alpha}\| g(T).
$$
Similarly, the ratio of the total cost to that of the optimum is at least
$$
  \frac{C_1(x,T)}{C_1(OPT_d,T)}
  \ge \frac{[T/(3M) - g(T)] 3\alpha\gamma\|1/\bar{\alpha}\|}{[T/(3M)]
  2\alpha\|1/\bar{\alpha}\|}
  = \frac{3}{2}\gamma\left(1 - \frac{3Mg(T)}{T}\right).
$$
If $g(T) = \Omega(T)$, then $R'_1(A)=\Omega(T)$.  Conversely,
if $g(T) = o(T)$, then for sufficiently large $T$, $3Mg(T)/T < 1/3$ and so
$\CR^\alpha_1(A) > \gamma$.

\section{Balancing Regret and the Competitive Ratio}\label{sec:tradeoff}
Given the above incompatibility, it is necessary to reevaluate the goals for algorithm design.
In particular, it is natural now to seek tradeoffs such as being able to obtain $\epsilon T$ regret for arbitrarily small $\epsilon$ while remaining
$O(1)$-competitive, or being $\log\log T$-competitive while retaining sublinear regret.

To this end, in the following we present a novel algorithm, Randomly Biased Greedy (RBG),
which can achieve simultaneous bounds on regret $R'_0$ and competitive ratio $CR_1$, when
the decision space $F$ is one-dimensional.  The one-dimensional setting is the natural starting point
for seeking such a tradeoff given that the proofs of the incompatibility results all focus
on one-dimensional examples and that the one-dimensional case has recently been of
practical significance, e.g. \cite{LWAT11}.  The algorithm takes a norm $N$ as its input:

\begin{simplealgorithm}[Randomly Biased Greedy, RBG($N$)] \label{alg:2-competitive}
\quad \\ Given a norm $N$, define $w^0(x)=N(x)$ for all $x$ and
$w^t(x)=\min_y\{w^{t-1}(y)+c^t(y)+N(x-y)\}$. Generate a random number $r$ uniformly in $(-1,1)$. For each time step $t$,
go to the state $x^t$ which minimizes $Y^t(x^t)=w^{t-1}(x^t)+ r N(x^t)$.
\end{simplealgorithm}

RBG is motivated by~\cite{CMP08}, and
makes very limited use of randomness -- it parameterizes its ``bias'' using a single random $r\in(-1,1)$.
It then chooses actions to greedily minimize its ``work function'' $w^t(x)$.

As stated, RBG performs well for the $\alpha$-unfair competitive ratio, but performs poorly for the regret.
Theorem~\ref{thm:balance} shows that RBG($\|\cdot\|$) is
$2$-competitive,\footnote{Note that this improves the best known competitive ratio for this setting from 3
(achieved by Lazy Capacity Provisioning) to 2.} and hence has at best linear regret.
However, the key idea behind balancing regret and competitive ratio is to run RBG with a
``larger'' norm to encourage its actions to change less.  This can make the coefficient of
regret arbitrarily small, at the expense of a larger (but still constant) competitive ratio.

\begin{theorem} \label{thm:balance}
For a SOCO problem in a one-dimensional normed space $\|\cdot\|$, running RBG($N$) with a one-dimensional
norm having $N(1) = \theta \|1\|$ as input (where $\theta \geq 1$)
attains an $\alpha$-unfair competitive ratio $CR_1^\alpha$ of $(1+\theta)/\min\{\theta,\alpha\}$
and a regret $R'_0$ of $O(\max\{{T}/{\theta}, \theta\})$.
\end{theorem}

Note that Theorem \ref{thm:balance} holds for the usual metrics of MTS and OCO, which are the ``most incompatible'' case since the
cost functions are mismatched (cf. Theorem~\ref{thm:impossibility}).  
Thus, the conclusion of Theorem \ref{thm:balance} still holds when $R_0$ or $R_1$ is considered in place of $R'_0$.

The best $CR_1^\alpha$, $1+1/\alpha$, achieved by RBG is obtained with $N(\cdot)=\alpha\|\cdot\|$.
However, choosing $N(\cdot)=\|\cdot\|/\epsilon$ for arbitrarily small $\epsilon$ gives
$\epsilon T$-regret at the cost of a larger $CR_1^\alpha$.
Similarly, if $T$ is known in advance, choosing $N(1)=\theta(T)$ for some increasing function achieves an $O(\theta(T))$
$\alpha$-unfair competitive ratio and $O(\max\{T/\theta(T),\theta(T)\})$ regret;
taking $\theta(T)=O(\sqrt{T})$ gives $O(\sqrt{T})$ regret, which is optimal for arbitrary convex costs \cite{Zinkevich03}.
If $T$ is not known in advance, $N(1)$ can increase in $t$, and bounds similar to those in Theorem \ref{thm:balance} still hold.

\subsubsection*{Proof of Theorem \ref{thm:balance}}
To prove Theorem \ref{thm:balance}, we derive a more general tool for
designing algorithms that simultaneously balance regret and the $\alpha$-unfair competitive ratio.
In particular, for any algorithm $A$, let the operating cost be $OC(A)=\sum_{t=1}^T c^t(x^{t+1})$ and the switching cost be
$SC(A)=\sum_{t=1}^T \|x^{t+1} - x^t\|$, so that $C_1(A) = OC(A)+SC(A)$.
Define $OPT_{N}$ to be the dynamic optimal solution under the norm
$N(1) = \theta \|1\|$ ($\theta \geq 1$) with $\alpha=1$.
The following lemma is proven in Appendix \ref{sec:proof.fakenorm}.

\begin{lemma} \label{LM:FAKENORM}
Consider a one-dimensional SOCO problem with norm $\|\cdot\|$ and an online algorithm $A$ which, when run with norm $N$,
satisfies $OC(A(N)) \le OPT_{N}+O(1)$ along with $SC(A(N)) \le
\beta OPT_{N}+O(1)$  with $\beta = O(1)$.  Fix a norm $N$ such that $N(1) = \theta \|1\|$
with $\theta \geq 1$.  Then $A(N)$ has $\alpha$-unfair competitive ratio $CR_1^\alpha(A(N)) = (1+\beta)\max\{\frac{\theta}{\alpha},1\}$
and regret $R'_0(A(N))=O(\max\{\beta T, (1+\beta)\theta\})$ for the original SOCO problem with norm $\|\cdot\|$.
\end{lemma}

Theorem \ref{thm:balance} then follows from the following lemmas, proven in Appendices \ref{sec:proof.rbg_operating} and~\ref{sec:proof.rbg_switching}.

\begin{lemma} \label{LM:RBG_OPERATING}
Given a SOCO problem with norm $\|\cdot\|$, $\esp{OC(RBG(N))} \leq OPT_{N}$.
\end{lemma}

\begin{lemma} \label{LM:RBG_SWITCHING}
Given a one-dimensional SOCO problem with norm $\|\cdot\|$,\hfill \hbox{} $\esp{SC(RBG(N))} \leq OPT_{N}/\theta$ with probability $1$.
\end{lemma}

\section{Concluding Remarks}\label{sec:conclusion}
This paper studies the relationship between regret and competitive ratio when applied to the class of SOCO problems.  It shows that these metrics, from the
learning and algorithms communities respectively, are fundamentally incompatible, in the sense that algorithms with sublinear regret must have infinite competitive
ratio, and those with constant competitive ratio have at least linear regret.  Thus, the choice of performance measure significantly affects the style of algorithm
design.  It also introduces a generic approach for balancing these competing metrics, exemplified by a specific algorithm, RBG.

There are a number of interesting directions that this work motivates.  In particular, the SOCO formulation is
still under-explored, and many variations of the formulation discussed here are still not understood.  For example,
is it possible to tradeoff between regret and the competitive ratio in bandit versions of SOCO?  More generally, the
message from this paper is that regret and the competitive ratio are incompatible within the formulation of SOCO.
It is quite interesting to try to understand how generally this holds.  For example, does the ``incompatibility result''
proven here extend to settings where the cost functions are random instead of adversarial, e.g., variations of SOCO such
as $k$-armed bandit problems with switching costs?

\section*{Acknowledgments}
This work was supported by NSF grants CNS 0846025 and DoE grant \textrm{DE-EE0002890}, along with the
Australian Research Council (ARC) grants FT0991594 and DP130101378. Katrina Ligett gratefully acknowledges
the generous support of the Charles Lee Powell Foundation.  Alan Roytman was partially supported by NSF
grants IIS-1065276, CCF-1016540, CNS-1118126, and CNS-1136174.

\bibliographystyle{plain}
\bibliography{soco}

\appendix
\section{Proof of Proposition \ref{P.OCOTOSOCO}}\label{sec:proof.ocotosoco}
Recall that, by assumption, $\|\triangledown c^t(\cdot)\|_2$ is bounded.  So, let us define $D$
such that $\|\triangledown c^t(\cdot)\|_2 \le D$.  Next, due to the fact that all norms are equivalent
in a finite dimensional space, there exist $m,M>0$ such that for every $x$, $m\|x\|_a \le \|x\|_b \le M\|x\|_a$.
Combining these facts, we can bound the switching cost incurred by an OGD algorithm as follows:
 \begin{align*}
 \sum_{t=1}^T \|x^{t}-x^{t-1}\| &\le M \sum_{t=1}^T \|x^{t}-x^{t-1}\|_2 \\
 & \le M \sum_{t=1}^T \eta_t \| \triangledown c^t(\cdot) \|_2 \\
 & \le MD \sum_{t=1}^T \eta_t.
 \end{align*}
 The second inequality comes from the fact that projection to a convex set under the Euclidean norm is nonexpansive, i.e., $\|P(x)-P(y)\|_2 \le \|x-y\|_2$.
 Thus, the switching cost causes an additional regret of $\sum_{t=1}^T \eta_t=O(\rho_1(T))$ for the algorithm, completing the proof.

\section{Proof of Lemma \ref{LM:RGBOUND}} \label{sec:proof.rgbound}
Recall that the oblivious adversary chooses $c^t(\cdot)=f^0_2$ for
$\esp{x^t}\le 1/2$ and $c^t(\cdot)=f^0_1$ otherwise, where $x^t$ is the (random) choice of the algorithm at round $t$.
Therefore,
\begin{align*}
C_0(A) &\geq \esp{\sum_{t=1}^T \begin{cases}
        a(1-x^t)+b & \textrm{if } \esp{x^t}\le 1/2\\
        a x^t+b & \textrm{otherwise}
    \end{cases} }\\
    &= \esp{ bT + a \sum_{t=1}^T \left(1/2+(1/2-x^t)\sgn(1/2-\esp{x^t})  \right) }\\
    &= bT + a \sum_{t=1}^T \left(1/2+(1/2-\esp{x^t})\sgn(1/2-\esp{x^t})  \right)\\
    &= (a/2+b)T + a \sum_{t=1}^T |1/2-\esp{x^t}|,
\end{align*}
where $\sgn(x)=1$ if $x>0$ and $-1$ otherwise.
The static  optimum is not worse than the scheme that sets $x^t=1/2$ for all $t$, which has total cost $(a/2+b)T + \|1/2\|$.
This establishes Equation~\eqref{eq:rgbound}.

The dynamic scheme which chooses $x^{t+1}=0$ if $c^{t}=f^0_1$ and
$x^{t+1}=1$ if $c^{t}=f^0_2$ has total $\alpha$-unfair
cost not more than $(b+\|\alpha\|)T$.
This establishes Equation~\eqref{eq:crbound}.

\section{Proof of Lemma~\protect\ref{th:minRegret}}\label{sec:CR1R1lemma}
\begin{proof}
We only consider the case that $\bar\alpha=1$; other cases are analogous.
We prove the contrapositive (that if $C_1(x; m) -C_1(OPT_s,m)\le \epsilon\|1\|$ for
all $m$, then $x^t \le \delta$ for all $t \in [1,\tau]$).
We consider the case that $x^t$ are non-decreasing; if not, the switching
and operating cost can both be reduced by setting $(x^t)' = \max_{t'\le t}
x^{t'}$.

Note that $OPT_s$ sets $x^t=0$ for all $t$, which implies $C_1(OPT_s,m)=am$, and that
$$
    C_1(x;m) = x^m\|1\| - a\sum_{i=1}^m x^i + am.
$$
Thus, we want to show that if $x^m\|1\| - a\sum_{i=1}^m x^i \le \epsilon$ for
all $m\le \tau$, then $x^t < \delta$ for all $t\in [1,\tau]$.

Define $f_i(\cdot)$ inductively by $f_1(y) = 1/(1-y)$, and
\begin{equation*}
    f_i(y) = \frac{1}{1 - y}
    		   \left(1 + y \sum_{j=1}^{i-1}f_j(y)\right).
\end{equation*}
If $y<1$, then $\{f_i(y)\}$ are increasing in $i$.
Notice that $\{f_i\}$ satisfy
\begin{equation*}
    f_m(y)(1-y) - y\sum_{i=1}^{m-1} f_i(y) = 1.
\end{equation*}
Expanding the first term gives that for any $\hat{\epsilon}$,
\begin{equation*}
    \hat{\epsilon}f_m(a/\|1\|)
	- \frac{a}{\|1\|} \sum_{i=1}^{m} \hat{\epsilon}f_i(a/\|1\|)
    = \hat{\epsilon}.
\end{equation*}
If for some $\hat{\epsilon}>0$,
\begin{equation}\label{eq:deltaC}
    x^m - \frac{a}{\|1\|} \sum_{j=1}^m x^j \le \hat{\epsilon}
\end{equation}
for all $m \le \tau$,
then by induction
$x^i \le \hat{\epsilon}f_i(a/\|1\|) \le \hat{\epsilon}f_\tau(a/\|1\|)$
for all $i\le \tau$, where the last inequality uses the fact that $a<\|1\|$
and hence $\{f_i(a/\|1\|)\}$ are increasing in $i$.

Observe that the left hand side of Equation~\eqref{eq:deltaC} is
$(C_1(x; m) -C_1(OPT_s,m))/\|1\|$.
Define $\epsilon=\hat{\epsilon} = \delta / (2 f_\tau(a/\|1\|))$.
Assuming we have $(C_1(x; m) -C_1(OPT_s,m)) \le \epsilon\|1\|$ for all $m$, then
Equation~\eqref{eq:deltaC} holds for all $m$, and thus
$x^t \le \hat{\epsilon}f_\tau(a/\|1\|) = \delta/2 < \delta$ for all $t \in
[1,\tau]$.
\end{proof}

\section{Proof of Lemma \ref{LM:FAKENORM}} \label{sec:proof.fakenorm}
We first prove the $\alpha$-unfair competitive ratio result.
Let $\hat{x}^1,\hat{x}^2,\dots, \hat{x}^T$ denote the actions chosen by algorithm $ALG$ when running on a normed space with $N = \|\cdot\|_{ALG}$
as input.  Let $\hat{y}^1,\hat{y}^2,\dots, \hat{y}^T$ be the actions chosen by the optimal
dynamic offline algorithm, which pays $\alpha$ times more for
switching costs, on a normed space with $\|\cdot\|$ (i.e., $OPT^{\alpha}_d$).
Similarly, let $\hat{z}^1,\hat{z}^2,\dots, \hat{z}^T$ be the actions chosen by the optimal solution on a normed space with $\|\cdot\|_{ALG}$,
namely $OPT_{\|\cdot\|_{ALG}}$ (without an unfairness cost).
Recall that we have $C_1(ALG) = \sum_{t=1}^T c^t(\hat{x}^{t+1}) + \|\hat{x}^{t+1} - \hat{x}^{t}\|$,
$OPT^{\alpha}_d = \sum_{t=1}^T c^t(\hat{y}^{t}) + \alpha\|\hat{y}^{t} - \hat{y}^{t-1}\|$,
and $OPT_{\|\cdot\|_{ALG}} = \sum_{t=1}^T c^t(\hat{z}^{t}) + \|\hat{z}^{t} - \hat{z}^{t-1}\|_{ALG}$.
By the assumptions in our lemma, we know that $C_1(ALG) \le (1+\beta) OPT_{\|\cdot\|_{ALG}} + O(1)$. Moreover,
\begin{eqnarray*}
OPT^{\alpha}_d &=&  \sum_{t=1}^T c^t(\hat{y}^{t}) + \alpha\|\hat{y}^{t} - \hat{y}^{t-1}\|\\
&\geq& \sum_{t=1}^T c^t(\hat{y}^{t}) +  \frac{\alpha}{\theta}\|\hat{y}^{t} - \hat{y}^{t-1}\|_{ALG} \geq
\frac{OPT_{\|\cdot\|_{ALG}}}{\max\{1,\frac{\theta}{\alpha}\}}.
\end{eqnarray*}
The first inequality holds since $\|\cdot\|_{ALG}=\theta\|\cdot\|$ with $\theta \ge 1$.  Therefore,
$C_1(ALG) \le (1+\beta)\max\{1,\frac{\theta}{\alpha}\} OPT^\alpha_d$.

We now prove the regret bound.  Let $d_{\max}$
denote the diameter of the decision space (i.e., the length of the interval).
Recall that $C_0(ALG)=\sum_{t=1}^T c^{t}(\hat{x}^t) + \|\hat{x}^t - \hat{x}^{t-1}\|$ and
$OPT_{s} = \min_x \sum_{t=1}^T c^t(x)$.
Then we know that $C_0(ALG) \le C_1(ALG) + D \sum_{t=1}^T{\|x^{t+1} - x^{t}\|} + \|d_{\max}\|$ for some constant $D$ by Equation~\eqref{eq:cost_for_offset}.
Based on our assumptions, we have $\sum_{t} c^t(\hat{x}^{t+1}) \le OPT_{\|\cdot\|_{ALG}} + O(1)$ and
$\sum_{t} \|\hat{x}^{t+1} - \hat{x}^{t}\| \le \beta OPT_{\|\cdot\|_{ALG}} + O(1)$.  For convenience, we let $E = D+1 = O(1)$.
Then $C_0(ALG)$ is at most:
\begin{align*}
\sum_{t=1}^T c^{t}(\hat{x}^{t+1}) + E\|\hat{x}^{t+1} - \hat{x}^{t}\| + \|d_{\max}\| + O(1)
&\leq (1+ E\beta)OPT_{\|\cdot\|_{ALG}} + \|d_{\max}\| + O(1) \\
&\leq (1+ E\beta)(OPT_{s}+\|d_{\max}\|_{ALG}) + \|d_{\max}\| + O(1).
\end{align*}
Therefore, we get that the regret $C_0(ALG) - OPT_{s}$ is at most
\begin{align*}
E\beta OPT_s + \|d_{\max}\|(1+ E(1+\beta)\theta) + O(1)
= O(\beta OPT_s + (1+\beta)\theta) = O(\max \{\beta OPT_s,(1+\beta)\theta\}).
\end{align*}

In the OCO setting, the cost functions $c^t(x)$ are bounded from below by 0 and are typically bounded from above by a value independent of $T$
(e.g., \cite{Herbster98,Littlestone94}), so that $OPT_{s} = O(T)$.
This immediately gives the result that the regret is at most $O(\max \{\beta T,(1+\beta)\theta\})$.

\section{Proof of Lemma \ref{LM:RBG_OPERATING}}\label{sec:proof.rbg_operating}
In this section, we argue that the expected operating cost of RBG (when evaluated under $\|\cdot\|$) with
input norm $N(\cdot) = \theta \|\cdot\|$, $\theta \geq 1$, is at most
the cost of the optimal dynamic offline algorithm under norm $N$ (i.e., $OPT_N$).
Let $M$ denote our decision space. Before proving this result, let us introduce a useful lemma.
Let $\hat{x}^1,\hat{x}^2,\dots, \hat{x}^{T+1}$ denote the actions chosen by RBG (similarly, let $x^1_{OPT},\ldots,x^{T+1}_{OPT}$ denote
the actions chosen by $OPT_N$).

\begin{lemma} \label{lm:wxt}
$w^t(\hat{x}^{t+1})=w^{t-1}(\hat{x}^{t+1})+c^t(\hat{x}^{t+1})$.
\end{lemma}
\begin{proof}
We know that for any state $x \in M$, we have $w^t(x)=\min_{y \in
M}\{w^{t-1}(y)+c^t(y)+\theta\|x-y\|  \}$.
Suppose instead $w^t(\hat{x}^{t+1})=w^{t-1}(y)+c^t(y)+\theta\|\hat{x}^{t+1}-y\|$ for some $y \neq \hat{x}^{t+1}$. Then
\begin{align*}
Y^{t+1}(\hat{x}^{t+1})
&= w^t(\hat{x}^{t+1})+\theta r\|\hat{x}^{t+1}\| \\
&= w^{t-1}(y)+c^t(y)+\theta\|\hat{x}^{t+1}-y\| +\theta r\|\hat{x}^{t+1}\|\\
&> w^{t-1}(y)+c^t(y) + \theta r\|y\|\\
&= Y^{t+1}(y),
\end{align*}
which contradicts $\hat{x}^{t+1} =\arg\min_{y \in M} Y^{t+1}(y)$. Therefore $w^t(\hat{x}^{t+1})=w^{t-1}(\hat{x}^{t+1})+c^t(\hat{x}^{t+1})$.
\end{proof}

Now let us prove the expected operating cost of RBG is at most the total cost
of the optimal solution, $OPT_N$:
\begin{align*}
Y^{t+1}(\hat{x}^{t+1})-Y^{t}(\hat{x}^{t})
&\geq Y^{t+1}(\hat{x}^{t+1})-Y^{t}(\hat{x}^{t+1})  \\
&= (w^t(\hat{x}^{t+1})+\theta r\|\hat{x}^{t+1}\|)  - (w^{t-1}(\hat{x}^{t+1})+\theta r\|\hat{x}^{t+1}\| ) \\
&= c^t(\hat{x}^{t+1}).
\end{align*}
Lemma \ref{LM:RBG_OPERATING} is proven by summing up the above inequality for
$t=1,\dots,T$, since  $Y^{T+1}(\hat{x}^{T+1}) \le Y^{T+1}(x_{OPT}^{T+1})$ and
$\esp{Y^{T+1}(x_{OPT}^{T+1})}=OPT_N$ by $\esp{r} =~0$.

Note that this approach also holds when the decision space $F \subset \mathbb{R}^n$ for $n>1$.

\section{Proof of Lemma \ref{LM:RBG_SWITCHING}}\label{sec:proof.rbg_switching}
To prove Lemma~\ref{LM:RBG_SWITCHING}, we make a detour and consider a version of the problem with a
discrete state space.  We first show that on such spaces the lemma holds for a discretization of RBG,
which we name DRBG.  Next, we show that as the discretization becomes finer, the solution (and hence
switching cost) of DRBG approaches that of RBG.  The lemma is then proven by showing that the optimal
cost of the discrete approximation approaches the optimal cost of the continuous problem.

To begin, define a discrete variant of SOCO where the number of states is finite as follows.
Actions can be chosen from $m$ states, denoted by the set $M = \{x_1,\ldots,x_m\}$, and the
distances $\delta=x_{i+1}-x_{i}$ are the same for all $i$.  Without loss of generality we define $x_1=0$.
Consider the following algorithm.

\begin{simplealgorithm}[Discrete RBG, DRBG($N$)] \label{alg:discrete} \hfill
\quad \\ Given a norm $N$ and discrete states $M=\{x_1,\ldots,x_m\}$, define $w^0(x)=N(x)$ and $w^t(x)=\min_{y\in
M}\{w^{t-1}(y)+c^t(y)+N(x-y)  \}$ for all $x \in M$. Generate a random number $r\in(-1,1)$. For each time
step $t$, go to the state $x^t$ which minimizes $Y^t(x^t)=w^{t-1}(x^t)+r N(x^t)$.
\end{simplealgorithm}

Note that DRBG looks nearly identical to RBG except that the states are discrete.  DRBG is introduced only for the proof and need never be implemented; thus we do not need to worry about the computational issues when the number of states $m$ becomes large.

\subsection{Bounding the Switching Cost of DRBG}
We now argue that the expected switching cost of DRBG (evaluated under the norm $\|\cdot\|$ and run with input norm
$N(\cdot) = \theta \|\cdot\|$) is at most the total cost of the optimal solution in the discrete system (under norm $N$).  We first prove a
couple of useful lemmas regarding the work function.  The first lemma
states that if the optimal way to get to some state $x$ at time $t$ is to come to state $y$ in the previous time step, incur the
operating cost at state $y$, and travel from state $y$ to state $x$, then in fact the optimal way to get to state $y$ at time $t$
is to come to $y$ at the previous time step and incur the operating cost at state $y$.
\begin{lemma}\label{lem:workstruct}
If $\exists x,y : w^t(x) = w^{t-1}(y) + c^t(y) + \theta\|x - y\|$, then $w^t(y) = w^{t-1}(y) + c^t(y)$.
\end{lemma}
\begin{proof}
Suppose towards a contradiction that $w^t(y) < w^{t-1}(y) + c^t(y)$.  Then we have:
\begin{align*}
w^t(y) + \theta\|x - y\| &< w^{t-1}(y) + c^t(y) + \theta\|x - y\| \\
  &= w^t(x) \leq w^t(y) + \theta\|x - y\|
\end{align*}
(since one way to get to state $x$ at time $t$ is to get to state $y$ at time $t$ and travel from $y$ to $x$).
This is a contradiction, which proves the lemma.
\end{proof}
The second lemma we show regarding the work function is as follows.
\begin{lemma}\label{lem:chain}
Suppose there is some state $x$ for which $w^t(x) = w^{t-1}(x) + c^t(x)$.  If $c^t(z) \geq c^t(x)$
for all $z \geq x$, then we have $w^t(z) \geq w^{t-1}(z) + c^t(x)$ for all $z \geq x$
(similarly, if $c^t(z) \geq c^t(x)$ for all $z \leq x$, then we have $w^t(z) \geq w^{t-1}(z) + c^t(x)$
for all $z \leq x$).
\end{lemma}
\begin{proof}
Assume without loss of generality that the entries in the cost vector satisfy $c^t(z) \geq c^t(x)$ for all
$z \geq x$.  Let $z$ be any state such that $z > x$, and assume towards a contradiction that
$w^t(z) < w^{t-1}(z) + c^t(x)$.  The optimal way to get to $z$ at time step $t$, $w^t(z)$, must
go through some point $j$ in the previous time step and incur the operating cost at $j$. 
If $j \geq x$, then we know
\begin{align*}
w^{t-1}(j) + c^t(x) + \theta\|z - j\| &\leq w^{t-1}(j) + c^t(j) + \theta\|z - j\| = w^t(z)\\
&< w^{t-1}(z) + c^t(x) \leq w^{t-1}(j) + \theta\|z - j\| + c^t(x),
\end{align*}
which cannot happen.  On the other hand, by Lemma~\ref{lem:workstruct}, if $j < x$, then we have
\begin{align*}
w^t(x) + \theta\|z - x\| &\leq w^t(j) + \theta\|1\||x - j| + \theta\|1\||z - x| \\
&= w^t(j) + \theta\|z - j\| = w^{t-1}(j) + c^t(j) + \theta\|z - j\| \\
&= w^t(z) < w^{t-1}(z) + c^t(x) \leq w^{t-1}(x) + \theta\|z-x\| + c^t(x),
\end{align*}
which cannot happen either.
\end{proof}

We now argue that, assuming the cost vectors are of a particular form, the algorithm can
only move from one state to another state (which are independent of the randomness $r$).
More specifically, at any particular time step, if the algorithm does ever move, it always moves from a unique state $x$
and it always moves to a unique state $y$ ($x$ and $y$ are independent of the randomness $r$ and hence remain the
same for all values of $r$ that cause the algorithm to move).
\begin{lemma}\label{lem:uniquestates}
Fix any time step $t$, and assume the entries in the cost vector $c^t$ are either $0$ or $\epsilon$ in each coordinate
(for a sufficiently small $\epsilon > 0$), and are either non-increasing or non-decreasing.  There exist states $x,y$
such that, for any $r$, we have the following guarantee.  At time $t$, we only have the following two possibilities
for this value of $r$:
\begin{enumerate}
\item The algorithm does not move.
\item The algorithm moves from state $x$ to state $y$.
\end{enumerate}
\end{lemma}
\begin{proof}
Fix a time step $t$, and assume without loss of generality the cost vector
$c^t = (0,\ldots,0,\epsilon,\ldots,\epsilon)$ (the case when the entries are non-increasing is symmetric).
Let $A = \{j : Y^{t+1}(j) = Y^t(j)\}$ and let $B = \{j : Y^{t+1}(j) = Y^t(j) + \epsilon\}$.
That is, $A$ is the set of states $j$ such that the values $Y^t(j)$ do not increase
(and in particular, the work function values $w^{t-1}(j)$ also do not
increase), and $B$ is the set of states $j$ such that the values $Y^t(j)$ do increase.
Note that sets $A$ and $B$ define a partition of the set of all states and are independent of $r$, since only an increase in
the work function value at a state $j$ can cause an increase in $Y^{t}(j)$ (note that the work function value is independent of $r$).
Moreover, by Lemma~\ref{lem:chain}, we know that sets $A$ and $B$ have the form $A = \{1,\ldots,i\},B=\{i+1,\ldots,m\}$ for
some $i$.  If set $B$ is empty, then the algorithm never moves at time step $t$ since at least one state's work function value must
increase for the algorithm to move (this is true for all $r$).  Moreover, if set $A$ is empty, then the algorithm also cannot move at time step $t$ since
at least one state's work function value must not increase (this is true for all $r$).
Hence, we assume that both sets $A$ and $B$ are non-empty, and moreover we assume this for all values of $r$. 

Now, fix a value for $r$, and consider the values $Y^t(j)$ for all states $j$ (the shape of $Y^t$ may be somewhat
arbitrary).  It is useful to understand how various values of $r$ affect the shape of $Y^t$.  As we increase the value for $r$,
the value of $Y^t(j)$ certainly increases for all states, but states $j$
which are farther to the right have the property that $Y^t(j)$ increases at a faster rate (and hence, states which are farther
to the left have a slower rate of increase).  Moreover, as we decrease the value for $r$, the value of $Y^t(j)$ decreases
for all states $j$, but states $j$ which are farther to the right have the property that $Y^t(j)$ decreases at a faster
rate (similarly, states farther to the left have a slower rate of decrease).  These properties hold due to the fact that
the dependence on $r$ for $Y^t(j)$ appears in the term $r\cdot N(j)$, and hence as we change $r$, larger values of $j$ have a larger impact
on the value of $Y^t(j)$ (since $N(j)$ is larger for such states $j$).

With these observations in mind, we again take $r$ to be any fixed value, and we also define $a_r = \arg\min_{j \in A} Y^t(j)$,
$b_r = \arg\min_{j \in B} Y^t(j)$ (recall that we assume $A$ and $B$ are non-empty for all values of $r$).  Note that
$a_r$ and $b_r$ may depend on the particular value of $r$, and moreover we always have $a_r < b_r$ for all values of $r$ (since
states in set $B$ are farther to the right).  In particular, the algorithm can only move from a state in $B$ to a state in $A$.
In addition, the global minimum value of $Y^t$ is precisely $\min\{Y^t(a_r),Y^t(b_r)\}$.
It is useful to note that, as we increase $r$, $a_r$ and $b_r$ may decrease (i.e., the minimum state in $A$ may move left
and the minimum state in $B$ may move left), while decreasing $r$ can cause $a_r$ and $b_r$ to increase.

Suppose that for every $r$ we have $Y^t(a_r) \neq Y^t(b_r)$.
This implies that either $Y^t(a_r) < Y^t(b_r)$ for all $r$ or $Y^t(a_r) > Y^t(b_r)$ for all $r$.  In other words,
it is impossible for there to exist $r_1,r_2$ such that $Y^t(a_{r_1}) < Y^t(b_{r_1})$ while $Y^t(a_{r_2}) > Y^t(b_{r_2})$.
If such a scenario existed, it would imply that there exists some value $r'$ such that $Y^t(a_{r'}) = Y^t(b_{r'})$ (i.e.,
a crossover point) due to continuity.  Hence, in the case that $Y^t(a_r) \neq Y^t(b_r)$ for all $r$, we must have that either $Y^t(a_r) < Y^t(b_r)$
for every $r$, or we must have $Y^t(a_r) > Y^t(b_r)$ for every $r$.  In either case, it is impossible for the algorithm to move
(i.e., for all values of $r$, the algorithm does not move).  To see why, consider the case when $Y^t(a_r) < Y^t(b_r)$ for every $r$.
This means that the state which achieves the global minimum of $Y^t$ (i.e., $\min\{Y^t(a_r),Y^t(b_r)\}$)
lies in $A$ for every value of $r$, and since the algorithm never moves from a state in $A$ after receiving the cost vector $c^t$, the first
case is done.  A similar argument can be made in the second case where for all values of $r$ we have $Y^t(a_r) > Y^t(b_r)$.  In particular,
although the global minimum lies in the set $B$ for every $r$, in the second case we know that for every $r$ we have $Y^t(b_r) < Y^t(a_r) \leq
Y^t(j)$ for every $j \in A$ (we can assume $\epsilon$ is small enough that the new global minimum after $c^t$ arrives still remains in $B$).

Hence, we assume there exists an $r$ such that $Y^t(a_r) = Y^t(b_r)$.  We define the state $y$ to be $a_r$,
and the state $x$ to be $b_r$.  Note that $x$ and $y$ are unique.  If there are ties for the minimum $Y^t$ value at the crossover
point in set $B$, we take $x$ to be the rightmost such point since states to the left in $B$ cannot be the minimum for smaller
values of $r$ (ties in set $A$ can be dealt with similarly when defining $y$).  Although $a_r$ and $b_r$ depend on $r$, we can claim uniqueness
due to the fact that $b_r > a_r$ and hence $Y^t(b_r)$ increases at a faster rate than $Y^t(a_r)$ as we increase $r$ and
$Y^t(b_r)$ decreases at a faster rate than $Y^t(a_r)$ as we decrease $r$.  Hence, let $r^*$ denote the unique value at which $Y^t(a_{r^*})$
and $Y^t(b_{r^*})$ meet.  Now that $x$ and $y$ are defined, let us see why the lemma holds.

Observe that the only values of $r$ for which the algorithm can move are precisely those when the algorithm is currently at the minimum
state in set $B$, namely $b_r$, and the values $Y^t(a_r) > Y^t(b_r)$ are really close together (in particular, increasing
$Y^t(b_r)$ by $\epsilon$ causes $Y^{t+1}(a_r) \leq Y^{t+1}(b_r)$ for a carefully chosen, sufficiently small $\epsilon$).
Consider the value $r^*$, so that $Y^t(x)$ is the minimum value in set $B$ and $Y^t(y)$ is the minimum value in set $A$.
Observe that for all larger values $r' \geq r^*$, the algorithm does not move since the global minimum is in set $A$ for such
values $r'$.

Now, consider a slightly smaller value $\tilde{r} < r^*$, which is sufficiently close to $r^*$ so that $Y^t(y)$ is
still the minimum value in set $A$ and $Y^t(x)$ is still the minimum value in set $B$, namely $y = a_{\tilde{r}}$ and $x = b_{\tilde{r}}$
(if $\tilde{r}$ is chosen too small, it is possible that $x$ and $y$ do not satisfy these properties).  Choose $\epsilon$ to be
sufficiently small so that, $Y^{t+1}(a_{\tilde{r}}) = Y^{t}(b_{\tilde{r}}) + \epsilon$.  Now, for all values $r' < \tilde{r}$,
the algorithm cannot possibly move from any state, since the gap between $Y^t(a_{r'})$ and $Y^t(b_{r'})$ is larger than
the gap between $Y^t(a_{\tilde{r}})$ and $Y^t(b_{\tilde{r}})$, and $\epsilon$ is sufficiently small that the global minimum
remains in $B$ after the cost vector arrives.  Now, consider values $r'$ such that $r^* > r' \geq \tilde{r}$.  It is not possible
for another state $j \in B$ to become the minimum state in $B$ for this range, by definition of how we chose $\tilde{r}$ and by
definition of $x$ (similar reasoning shows that no other state $j \in A$ can be the minimum state in $A$ for this range).
Hence, every time the algorithm moves, it goes from state $x$ to state $y$.  In particular, $x$ remains the minimum state in set $B$
and $y$ remains the minimum state in set $A$ for this range, and the algorithm moves from state $x$ to state $y$ for all
$r'$ in this range.
\end{proof}

We now prove the main lemma. Let $SC^t=\sum_{i=1}^t \|x^i-x^{i-1}\|$ denote the total switching cost incurred by DRBG up until time $t$, and
define the potential function $\phi^t = \frac{1}{2\theta}(w^t(x_1) + w^t(x_m)) - \frac{\|x_m-x_1\|}{2}$.  Then we can show the following lemma.
\begin{lemma}\label{lem:potential}
For every time step $t$, $\esp{SC^t} \leq \phi^t$.
\end{lemma}
\begin{proof}
We prove this lemma by induction on $t$.  At time $t = 0$, clearly it is true since the left hand side $\esp{SC^0} = 0$,
while the right hand side
$\phi^0 = \frac{1}{2\theta}(w^0(x_1) + w^0(x_m)) - \frac{\|x_m-x_1\|}{2} = \frac{1}{2\theta}(0 + \theta\|x_m-x_1\|) - \frac{\|x_m-x_1\|}{2} = 0$.
We now argue that at each time step, the increase in the left hand
side is at most the increase in the right hand side.

Since the operating cost is convex, it is non-increasing until some point $x_{min}$ and then non-decreasing over the set $M$.
We can imagine our cost vector arriving in $\epsilon$-sized increments as follows.  We imagine sorting the cost values so that
$c^t(i_1) \leq c^t(i_2) \leq \cdots \leq c^t(i_m)$, and then view time step $t$ as a series of smaller time steps where we apply a cost
of $\epsilon$ to all states for the first ${c^t(i_1)}/{\epsilon}$ time steps, followed by applying a cost of $\epsilon$ to
all states except state $i_1$ for the next $({c^t(i_2) - c^t(i_1)})/{\epsilon}$ time steps (if each such cost vector's entries
strictly decrease at some point and then strictly increase at some point, we split the vector into two vectors which
add up to the original, one of which is non-increasing and the other of which is non-decreasing), etc.,
where $\epsilon$ has a very small value.  If adding these $\epsilon$-sized cost vectors would cause us to exceed the
original cost $c^t(i_k)$ for some $k$, then we just use the residual $\epsilon'<\epsilon$ in the last round in
which state $i_k$ has non-zero cost.  Eventually, these $\epsilon$-sized cost vectors add up precisely to the original
cost vector $c^t$.  Under these new cost vectors, the algorithm's switching cost does not get better (and the optimal
solution does not get worse).  If the left hand side does not increase at all from time step $t-1$ to $t$, then
the lemma holds (since the right hand side can only increase).

Our expected switching cost is the probability that the algorithm moves multiplied by the distance moved.
Suppose the algorithm is currently in state $x$.  Observe that, by Lemma~\ref{lem:uniquestates}, there is
only one state the algorithm could be moving from (state $x$) and only one state $y$ the algorithm could be
moving to, both of which do not depend on the randomness $r$ (we can choose $\epsilon$ to be sufficiently
small in order to guarantee this).  Moreover, we would never move to a state where the work function increases
by $\epsilon$.  First we consider the case~$x \geq x_{min}$.

The only reason we would move from state $x$ is if $w^t(x)$ increases from the previous time step,
so that we have $w^t(x) = w^{t-1}(x) + \epsilon$.  By Lemma~\ref{lem:chain}, we know
$w^t(z) = w^{t-1}(z) + \epsilon$ for all $z \geq x$.
Hence, we can conclude a couple of facts.  The state $y$ we move to cannot be such that $y \geq x$.
Moreover, we also know that $w^t(x_m) = w^{t-1}(x_m) + \epsilon$ (since $x_m \geq x$).  Notice that for us to move
from state $x$ to state $y$, the random value $r$ must fall within a very specific range.  In particular, we must have
$Y^{t}(x) \leq Y^{t}(y)$ and $Y^{t+1}(y) \leq Y^{t+1}(x)$:
\begin{align*}
    &(w^{t-1}(x) + \theta r \|1\|x \leq w^{t-1}(y) + \theta r \|1\|y ) \\
    &\: \wedge \: (w^t(y) + \theta r \|1\|y \leq w^t(x) + \theta r \|1\|x ) \\
  \Longrightarrow &
    w^{t-1}(y) - w^{t-1}(x) - \epsilon \leq w^t(y) - w^t(x) \\
    & \leq \theta r \|x-y\| \leq w^{t-1}(y) - w^{t-1}(x).
\end{align*}
This means $r$ must fall within an interval of length at most
${\epsilon}/{\theta\|x - y\|}$.  Since $r$ is chosen from an interval of length 2,
this happens with probability at most ${\epsilon}/({2\theta\|x - y\|})$.  Hence, the increase
in our expected switching cost is at most $\|x - y\|\cdot{\epsilon}/({2\theta\|x - y\|})  = {\epsilon}/{2\theta}$.  On the other
hand, the increase in the right hand side is $\phi^t - \phi^{t-1} = \frac{1}{2\theta}(w^t(x_1) - w^{t-1}(x_1) + w^t(x_m) - w^{t-1}(x_m))
\geq {\epsilon}/{2\theta}$ (since $w^t(x_m) = w^{t-1}(x_m) + \epsilon)$.  The case when $x < x_{min}$ is symmetric.  This
finishes the inductive claim.
\end{proof}

Now we prove the expected switching cost of DRBG is at most the total cost of the optimal solution for the discrete problem.

By Lemma~\ref{lem:potential}, for all times $t$ we have $\esp{SC^t} \leq \phi^t$.  Denote by $OPT^t$ the
optimal solution at time $t$ (so that $OPT^t = \min_x w^t(x)$ and $OPT^T =
OPT_N$).  Let $x^* = \arg\min_x w^t(x)$
be the final state which realizes $OPT^t$ at time $t$.  We have, for all times $t$:
\begin{align*}
\esp{SC^t} \leq& \phi^t = \frac{1}{2\theta}(w^t(x_1) + w^t(x_m)) - \frac{\|x_m-x_1\|}{2}\\
\leq& \frac{1}{2\theta}(w^t(x^*) + \theta\|x^* - x_1\| + w^t(x^*) + \theta\|x_m  - x^*\|)
- \frac{\|x_m-x_1\|}{2} = \frac{1}{\theta}OPT^t.
\end{align*}
In particular, the equation holds at time $T$, which gives the bound.

\subsection{Convergence of DRBG to RBG}
In this section, we are going to show that if we keep splitting the state spacing $\delta$,
then the output of DRBG, which is denoted by $x_D^t$, converges to the output of RBG, which is denoted by $x_C^t$.
\begin{lemma} \label{lm:convergence}
Consider a SOCO with $F=[x_L,x_H]$.
Consider a sequence of discrete systems such that the state spacing $\delta \to 0$ and for each
system, $[x_1,x_m]=F$. Let $x_{i}$ denote the output of DRBG in the $i^{th}$ discrete system, and
$\hat{x}$ denote the output of RBG in the continuous system. Then the sequence $\{x_i\}$  converges
to $\hat{x}$ with probability $1$ as $i$ increases, i.e., for all $t$, $\lim_{i\to\infty} |x_i^t-\hat{x}^t|=0$
with probability $1$.
\end{lemma}
\begin{proof}
To prove the lemma, we just need to show that $x_i$ converges pointwise to $\hat{x}$ with probability~$1$.
For a given $\delta$, let $Y_D^t$ denote the function $Y^t$ used by DRBG in the discrete system (with feasible set
$M=\{x_1, \dots,x_m\}\subset F$) and $Y_C^t$ denote the function $Y^t$ used by RBG in the continuous
system (with feasible set $F$).  The output of DRBG and RBG at time $t$ are denoted by $x_D^t$ and $x_C^t$ respectively.
The subsequence on which $|x_C^t-x_D^t|\le 2\delta$ clearly has $x_D^t$ converge to $x_C^t$.  Now consider the
subsequence on which this does not hold.  For each such system, we can find an $\bar{x}_C^t \in \{x_1, \dots,x_m\}$
satisfying $|\bar{x}_C^t - x_C^t| < \delta$ (and thus $|\bar{x}_C^t - x_D^t| \ge \delta$) such that
$Y_C^t(\bar{x}_C^t) \le Y_C^t(x_D^t)$, by the convexity\footnote{The minimum of a convex function over a convex set is convex.
Thus, by definition, $w^t$ is a convex function by induction, and hence $Y_C^t$ is convex as well.} of $Y_C^t$. Moreover,
$Y_D^t(x_D^t) \le Y_D^t(\bar{x}_C^t)$ and $Y_C^t(x_D^t) \le Y_D^t(x_D^t)$. So far, we have only rounded the $t^{th}$ component.
Now let us consider a scheme that rounds to the set $M$ all components $\tau<t$ of a solution to the continuous problem.

For an arbitrary trajectory $x=(x^t)_{t=1}^T$, define a sequence $x_R(x)$ with $x_R^\tau \in \{x_1, \dots,x_m\}$ as follows.
Let $l=\max\{k:x_k\le x^\tau\}$. Set $x_R^\tau$ to $x_l$ if $c^\tau(x_l) \le c^\tau(x_{l+1})$ or $l=m$, and $x_{l+1}$ otherwise.
This rounding increases the switching cost by at most $2\theta\|\delta\|$ for each timeslot.  If $l=m$ then the operating cost
is unchanged.  Next, we bound the increase in operating cost when $l<m$.

For each timeslot $\tau$, depending on the shape of $c^\tau$ on $(x_l,x_{l+1})$, we may have two cases: (i) $c^\tau$ is monotone;
(ii) $c^\tau$ is non-monotone.  In case (i), the rounding does not make the operating cost increase for this timeslot. Note that
if $x_C^\tau \in \{x_L,x_H\}$ then for all sufficiently small $\delta$, case (ii) cannot occur, since the location of the minimum
of $c^\tau$ is independent of $\delta$.  We now consider case (ii) with $x_C^\tau \in (x_L,x_H)$.  Note that there must be a finite
left and right derivative of $c^\tau$ at all points in $(x_L,x_H)$ for $c^\tau$ to be finite on $F$.  Hence, these derivatives must
be bounded on any compact subset of $(x_L,x_H)$.  Since $x_C^\tau \in (x_L, x_H)$, there exists a set $[x_L',x_H'] \subset (x_L,x_H)$
independent of $\delta$ such that, for sufficiently small $\delta$, we have $[x_l,x_{l+1}] \subset [x'_L,x'_H]$.  Hence, there exists
an $H^\tau$ such that, for sufficiently small $\delta$, the gradient of $c^\tau$ is bounded by $H^\tau$ on $[x_l,x_{l+1}]$.  Thus, for
sufficiently small $\delta$, the rounding makes the operating cost increase by at most $H^\tau \delta$ in timeslot $\tau$.

Define $H=\max_\tau\{H^\tau\}$.  If we apply this scheme to the trajectory
which achieves $Y_C^t(\bar{x}_C^t)$, we get a decision sequence in the discrete system with $cost+r\theta\|\bar{x}_C^t\|$ not more than $Y_C^t(\bar{x}_C^t) + (H
\delta +2\theta\|\delta\|)t$ (by the foregoing bound on the increase in costs) and not less than $Y_D^t(\bar{x}_C^t)$ (because the solution of $Y_D^t(\bar{x}_C^t)$
minimizes $cost+r\theta\|\bar{x}_C^t\|$). Specifically, we have $Y_D^t(\bar{x}_C^t) \le Y_C^t(\bar{x}_C^t) + (H \delta +2\theta\|\delta\|)t$. Therefore,
\begin{align*}
      Y_C^t(\bar{x}_C^t)
	\le Y_C^t(x_i^t) = Y_C^t(x_D^t)
	\le Y_D^t(x_D^t)
	\le Y_D^t(\bar{x}_C^t)
	\le Y_C^t(\bar{x}_C^t) + (H \delta +2\theta\|\delta\|)t.
\end{align*}
\noindent Notice that the gradient bound $H$ is independent of $\delta$, and so $(H \delta +2\theta\|\delta\|)t \to 0$
as $\delta\to0$.  Therefore, $| Y_C^t(x_i^t) - Y_C^t(\bar{x}_C^t)|$ converges to $0$ as $i$ increases.

Independent of the random choice $r$, the domain of $w_C^t(\cdot)$ can be divided into countably many non-singleton intervals
on which $w_C^t(\cdot)$ is affine, joined by intervals on which it is strictly convex.  Then $Y_C^t(\cdot)$ has a unique
minimum unless $-r$ is equal to the slope of one of the former intervals, since $Y_C^t(\cdot)$ is convex.
Hence, it has a unique minimum with probability one with respect to the choice of $r$.

Hence, with probability one, $x_C^t$ is the unique minimum of $Y_C^t$. To see that $Y_C^t(\cdot)$ is continuous at
any point $a$, apply the squeeze principle to the inequality $w_C^t(a) \le w^t_C(x) + \theta\|x - a\| \le w^t_C(a) + 2\theta\|x - a\|$,
and note that $Y_C^t(\cdot)$ is $w^t(\cdot)$ plus a continuous function.  The convergence of $|\bar{x}_C^t - x_C^t|$
then implies $| Y_C^t(\bar{x}_C^t) - Y_C^t(x_C^t)|\to0$ and thus $| Y_C^t(x_i^t) - Y_C^t(x_C^t)|\to0$, or
equivalently $Y_C^t(x_i^t) \to Y_C^t(x_C^t)$.  Note that the restriction of $Y_C^t$ to $[x_L,x_C^t]$ has a well-defined
inverse $Y^{-1}$, which is continuous at $Y_C^t(x_C^t)$.  Hence, for the subsequence of $x_i^t$ such that
$x_i^t\le x_C^t$, we have $x_i^t = Y^{-1}(Y_C^t(x_i^t)) \to Y^{-1}(Y_C^t(x_C^t)) = x_C^t$.  Similarly,
the subsequence such that $x_i^t \ge x_C^t$ also converges to $x_C^t$.
\end{proof}

\subsection{Convergence of OPT in Discrete Systems}
To show that the competitive ratio holds for RBG, we also need to show that the optimal costs converge to those of the continuous system.

\begin{lemma} \label{lm:convergenceOPT}
Consider a SOCO problem with $F=[x_L,x_H]$.
Consider a sequence of discrete systems such that the state spacing $\delta \to 0$ and for each system, $[x_1,x_m]=F$. Let $OPT^i_D$ denote the optimal cost in the $i^{th}$ discrete system, and $OPT_C$ denote the optimal cost in the continuous system (both under the norm $N$). Then the sequence $\{OPT^i_D\}$  converges to $OPT_C$ as $i$ increases, i.e., $\lim_{i\to\infty} |OPT^i_D-OPT_C|=0$.
\end{lemma}

\begin{proof}
We can apply the same rounding scheme in the proof of Lemma \ref{lm:convergence} to the solution vector of $OPT_C$ to get a discrete output with total cost bounded by
$OPT_C+(H \delta +2\theta\|\delta\|)T$, thus
$$OPT^i_D \le OPT_C+(H \delta +2\theta\|\delta\|)T.$$
Notice that the gradient bound $H$ is independent of $\delta$ and so $(H \delta +2\theta\|\delta\|)T\to 0$ as $\delta\to0$. Therefore, $OPT^i_D$ converges to $OPT_C$ as $i$ increases.
\end{proof}

\end{document}